\documentclass[11pt]{article} % use larger type; default would be 10pt
\usepackage[utf8]{inputenc} % set input encoding (not needed with XeLaTeX)
\usepackage{geometry} % to change the page dimensions
\usepackage{url}
\usepackage{graphicx} % support the \includegraphics command and options
\usepackage[export]{adjustbox}
\usepackage{booktabs} % for much better looking tables
\usepackage{array} % for better arrays (eg matrices) in maths
\usepackage{paralist} % very flexible & customisable lists (eg. enumerate/itemize, etc.)
\usepackage{verbatim} % adds environment for commenting out blocks of text & for better verbatim
\usepackage{subfig} % make it possible to include more than one captioned figure/table in a single float
\usepackage{amsmath}
\usepackage{mathrsfs}
\usepackage{amsfonts}
\usepackage{amsthm}
\usepackage{amssymb}
\usepackage{bbm}
\usepackage{lipsum}                     % Dummytext
\usepackage{xargs}                      % Use more than one optional parameter in a new commands
\usepackage[colorinlistoftodos,prependcaption,textsize=tiny]{todonotes}
\newcommandx{\unsure}[2][1=]{\todo[linecolor=red,backgroundcolor=red!25,bordercolor=red,#1]{#2}}
\newcommandx{\change}[2][1=]{\todo[linecolor=blue,backgroundcolor=blue!25,bordercolor=blue,#1]{#2}}
\newcommandx{\info}[2][1=]{\todo[linecolor=OliveGreen,backgroundcolor=OliveGreen!25,bordercolor=OliveGreen,#1]{#2}}
\newcommandx{\improvement}[2][1=]{\todo[linecolor=Plum,backgroundcolor=Plum!25,bordercolor=Plum,#1]{#2}}
\newcommandx{\thiswillnotshow}[2][1=]{\todo[disable,#1]{#2}}

\usepackage[noadjust]{cite}

\usepackage{mathtools}
\newcommand\myeq{\stackrel{\mathclap{\normalfont\mbox{def}}}{=}}

\usepackage{fancyhdr} % This should be set AFTER setting up the page geometry
\usepackage{sectsty}
\allsectionsfont{\sffamily\mdseries\upshape} % (See the fntguide.pdf for font help)
\usepackage[nottoc,notlof,notlot]{tocbibind} % Put the bibliography in the ToC
\usepackage[titles,subfigure]{tocloft} % Alter the style of the Table of Contents

\theoremstyle{plain}
\newtheorem{thm}{Theorem}
\numberwithin{thm}{subsection}
\newtheorem{lem}[thm]{Lemma}

\theoremstyle{definition}
\newtheorem{defn}[thm]{Definition}
\newtheorem{example}[thm]{Example}

\theoremstyle{remark}
\newtheorem{rem}[thm]{Remark}

\title{Thiele's Differential Equation Based on Markov Jump Processes with Non-countable State Space}
\author{
Emmanuel Coffie\footnote{Department of Mathematics and Statistics, University of Strathclyde, Glasgow, G1 1XH, UK. Email: emmanuel.coffie@strath.ac.uk.} \\
  \and
Sindre Duedahl\footnote{Danske Bank, N-0250, Aker Brygge, Oslo, Norway.  Email: sidu@danskebank.com.}\\
  \and
Frank Proske\footnote{Department of Mathematics, University of Oslo, N-0316, Blindern, Oslo, Norway. Email: proske@math.uio.no.}\\
}
\begin{document}
\maketitle
\begin{abstract}
In modern life insurance, Markov processes in continuous time on a finite or at least countable state space have been over the years an important tool for the modelling of the states of an insured. Motivated by applications in disability insurance, we propose in this paper a model for insurance states based on Markov jump processes with more general state spaces. We use this model to derive a new type of Thiele's differential equation which e.g. allows for a consistent calculation of reserves in disability insurance based on two-parameter continuous time rehabilitation rates.

\end{abstract}

\section{Introduction}
Over the years, finite-state Markov chains have played a significant role in multi-state modelling of life insurance risks. Typical insurance applications of finite-state Markov chains pertain to e.g. endowment insurance, life annuities or pension contracts. In fact, there is a rich literature devoted to finite- or countable-state Markov chain modelling of life insurance risks with respect to calculation of reserves. For instance, Henriksen et al. in \cite{Henriksen} employ the finite-state Markov chain framework to model (in addition to insurance risk) behavioural risk in the sense of e.g. surrender or free policy risk and examine the effects of such types of risks on prospective reserves. 
\par
On the other hand, Norberg \cite{Norberg} considers the case of a force of interest modelled by a time-continuous homogeneous Markov process with finite-state space and applies this model to the computation of prospective reserves for some standard insurance policies. As for a variety of other important applications of finite-state (or countable-state) Markov processes or chains to issues as e.g. unit-linked insurance policies in life insurance, we refer the reader to \cite{Koller} or \cite{moller2007market} and the references therein. 
\par
Despite of the wide applicability of Markov processes with countable state spaces in insurance risk modelling for reserve calculations, such processes may not be sufficient to consistently describe crucial life insurance risks associated with certain modern life insurance policies. For example, life insurance risks in connection with disability insurance based on two-parameter continuous-time rehabilitation rates or "random spouse" contracts cannot be adequately modelled by Markov processes with countable state spaces, but by those on more general state spaces. 
\par
The application of Markov processes with more general state spaces could be in particular relevant for immunocompromised policyholders, who have a high exposure to disease recidivism over time and who need to earn benefits whilst undergoing recovery, rehabilitation or medical treatment.
\par
In this paper, we use Markov jump processes on more general spaces to model insurance risks and to establish a new type of Thiele's differential equation for the computation of insurance reserves. 
\par
Our paper is organised as follows: In Subsections 1.1 and 1.2, we introduce the mathematical setting of this article and some insurance notation needed later on. In Section 2, we derive Thiele's differential equation in the framework of Markov jump processes on more general spaces. Finally, in Section 3, we discuss some insurance policies which necessitate and justify the use of Markov processes with non-countable state spaces in risk modelling. Further, an example of numerical implementation is presented.
\subsection{Mathematical preliminaries}
In this Subsection, we pass in review some mathematical notions and results which we will need throughout the paper. See e.g. \cite{ethier2009markov} or \cite{blumenthal2007markov} as for results on Markov processes. 
\begin{defn}
Given a Polish space $S$, let $D(S)$ be the space of càdlàg functions from the interval $[0,\infty)$ into $S$ (i.e. the space of functions $f:[0,\infty) \to S$, which are right continuous with existing left sided limits). We also use the symbol $\mathscr{S}$ for the Borel $\sigma$-algebra on $S$. Also, define $\mathrm{PC}([0,\infty), S)$ to be the subpace of $D(S)$ determined by the additional requirement that functions are piecewise-constant, i.e. constant on the half-open intervals between jump discontinuities.
\end{defn}
In the next definition, we introduce the \emph{regular insurance model} in our settings. See 
\cite{Koller} in the case of finite-state Markov chains.

\begin{defn}\label{defn1}
A regular insurance model consists of the following objects:
\begin{enumerate}
\item A measurable space $(S, \mathscr{S})$ called the \emph{state space}, where $S$ is Polish and $\mathscr{S}$ is the Borel $\sigma$-algebra. 
\item A filtered probability space $(\Omega, \mathscr{F}, \{\mathscr{F}_t\}_{t\geq 0}, P)$.
\item A \emph{kernel of positive measure}, i.e. a map 
\[
[0,T] \times S \times \mathscr{S} \ni (t,x,B) \mapsto q_t(x, B) \in \mathbb{R}
\]
such that for every $t \in [0,T], B\in \mathscr{S}$, $x \mapsto q_t(x, B)$ is $\mathscr{S} - \mathscr{B}(\mathbb{R})$-measurable, and for every $t\in [0,T], x\in S$, $B \mapsto q_t(x, B)$ is a positive measure.

\item \label{condprob}
A Markov jump process on $S$, i.e. $X:\Omega \rightarrow D_T(S)$ whose paths are almost surely in $\mathrm{PC}([0, \infty), S)$, such that $X$ is $\mathscr{F}_t$-adapted and has the Markov property with respect to $\mathscr{F}_t$ and $P$, and such that  $q$ is the jump intensity function, i.e. for $x\in S$, 
\begin{equation}
P_{t,t+h}(x, B) = q_t(x, B)h + o(h) 
\end{equation}
for h $\searrow $ 0, where $P_{t,s}(x, B)$ is a transition function of $X$, i.e. 
\[
\label{condprob}
P_{t, s}(x, B) \triangleq P[X_s \in B | X_t = x]. 
\]

We use the notation $X_t$ for the random variable given by $X_t(\omega) = X(\omega)(t)$ for $t\in [0, \infty), \omega \in \Omega$.
 \item
 
A measurable function $B:[0,\infty) \times S \rightarrow \mathbb{R}$ which is of bounded variation (BV) in the first variable.
\item
A BV function $b:[0,\infty) \times S \times S \rightarrow \mathbb{R}$.

\end{enumerate}
\end{defn}

\begin{rem}
By the conditional probability expression in item \ref{condprob} above, we mean the following: 
Since $X_t$ is a random map into a Polish space there exist (see \cite{karshr1}, pp. 84-85 and references cited there) \emph{regular conditional probabilities}, i.e. maps $\nu_t: S\times \mathscr{F} \rightarrow [0, 1]$ such that 
\begin{enumerate}
\item $S\ni g\mapsto \nu_t(g, \Gamma) $ is measurable for every $ \Gamma \in \mathscr{F}$.\\
\item $\mathscr{F} \ni \Gamma \mapsto \nu_t(g, \Gamma)$ is a probability measure for each $g\in S$, and\\
\item $P(A\cap X_t^{-1}(B)) = \int_{h \in B}\nu_t(h, A)P(X_t^{-1}(dh))$.
\end{enumerate}
From the third equality it follows for all measurable maps $Y:\Omega\rightarrow S$ and $f\in L^1(PY^{-1})$  that 
\[
E(f(Y)|X_t) = \int_{h \in S}f(h)\nu_t(X_t, Y^{-1}(dh)), 
\] $PX_t^{-1}$-a.s. To see this, let $\phi:S\rightarrow \mathbbm{R}$ be measurable and assume moreover that $\phi, f$ both have finite image, i.e.  

\[
\phi = \sum_{j=1}^d \alpha_j\mathbbm{1}_{A_j}, \hspace{1cm} f = \sum_{j = 1}^v \beta_j \mathbbm{1}_{B_j}\\\\
\]
then observe that
\[
\begin{split}
E(\phi(X_t)f(Y)) &= \sum_{i=1}^d \sum_{j = 1}^v\alpha_i\beta_jE(\mathbbm{1}_{A_i}(X_t)\mathbbm{1}_{B_j}(Y))\\
&=\sum_{i=1}^d \sum_{j = 1}^v\alpha_i\beta_jP(Y^{-1}(B_j)\cap X_t^{-1}(A_i))\\
&=\sum_{i=1}^d\alpha_i \int_{h\in A_i} \left(\sum_{j=1}^v \nu_t(h, Y^{-1}(B_j))\right) P(X_t^{-1}(dh))\\
&=\sum_{i=1}^d\alpha_i \int_{h\in A_i} \int_{g\in S}f(g)\nu_t(h, Y^{-1}(dg))P(X_t^{-1}(dh))\\
&=E\left(\phi(X_t)\int_{g\in S}f(g)\nu(X_t, Y^{-1}(dg))\right)
\end{split}
\] 
We then infer the general case of $L^1$ functions by approximation with simple functions and continuity of the expectation functional (or alternatively the monotone class theorem). Or see e.g. \cite{ganssler2013} regarding the above relation.

In particular, we have that
\begin{equation}
\label{condint}
E(f(X_s)|X_t = g) = \int_{h\in S}f(h)\nu_t(g, X_s^{-1}(dh)).
\end{equation}

\end{rem}
\begin{rem}
\label{remark1}
In the sequel, the symbol $B(t, g)$, which we sometimes denote by $B_g(t)$ should be interpreted as the accumulated payment stream up to time $t$ if $X_s= g$ for all $s\leq t$. Since it is assumed that payments depend only on the time and state at each moment, the accumulated payments over an interval $I$ is given by $\int_I dB_{X_t}(t)$. 

The symbol $b(t, g, h)$, which we also denote by $b_{gh}(t)$, should be interpreted as an immediately incurred payment at the time $t$ of a transition from $g$ to $h$, i.e. the sum of all payments from such events over an interval $I$ is given by 
\[
\int_I b_{gh}(t)dN_{gh}(t),
\]

where $N_{gh}(t)$ is the number of jumps performed by $X$ from $g$ to $h$ up to time $t$, in other words:
\[
N_{gh}(t) =
 \begin{cases} 
      \#\{s|s\leq t,  X_{t-} = g, X_t = h\}, g \neq h \\
0, g = h
        \end{cases}
\]
\end{rem}

\begin{rem}
The assumptions made here imply that a Markov process having a given function as its jump intensity can always be constructed, see \cite{Eb15}. The construction involves decomposing the problem in two parts: The \emph{jumping times} 
\[
J_1, J_2, \dots
\]
 and the random sequence of occupied states 
 \[
 Y_0, Y_1, \dots
 \]

such that $X_t = Y_i$ for $t\in [J_i, J_{i+1}]$.

The conditional distribution of the next jumping time is given by the \emph{survival function}:

\begin{equation}
P(J_{i+1} - J_i >t | J_1, \dots, J_i, Y_0, \dots, Y_i) = e^{-\int_0^t\lambda_{J_i+s}(Y_i)ds}, 
\end{equation}

where $\lambda$ is the \emph{total jump rate} given by 
\[
\lambda_t(x) = q_t(x, S \setminus \{x\}).
\]

It will also be the case that 
\[
P[Y_i\in B|J_1, \dots, J_i, Y_0, \dots Y_{i-1}] = \pi_{J_i}(Y_{i-1}, B),
\]

where $\pi_t$ is a family of probability measures parametrized by $S$, given by:

\[
\pi_t(x, B) = \frac{q_t(x, B)}{\lambda_t(x)}.
\]

\end{rem}

\subsection{Transition probabilities}

We recall the  \emph{Kolmogorov-Chapman} equation and the \emph{backward and forward Kolmogorov} equations from Markov process theory. See e.g. \cite{ethier2009markov} or \cite{blumenthal2007markov}.

\begin{lem}
For $t < u < s$, $g \in S$ and $\Gamma \in \mathscr{S}$,
\begin{enumerate}[i)]
\item $P_{t,t}(g, \Gamma) = \mathbbm{1}_\Gamma(g) \label{kolmtt}$.\\
\item $P_{t,s}(g,\Gamma) = \int_{S} P_{t,u}(g,dh)P_{u,s}(h,\Gamma)\label{chapman}$.\\
\item $\frac{\partial}{\partial t}P_{t, s}(g, \Gamma) = \lambda_t(g)P_{t,s}(g, \Gamma) - \int_{S \setminus\{g\}}q_t(g, dh)P_{t,s}(h,\Gamma)\label{backward}$.\\
\item $\frac{\partial}{\partial s}P_{t, s}(g, dh) = -P_{t, s}(g, dh)\lambda_s(h) + \int_{S \setminus\{h\}}P_{t, s}(g, dk)q_s(k, dh).$\label{forward}
\end{enumerate}
%
%\begin{proof}
%\ref{kolmtt}) is trivial. \\
%Let $t<u<s$ and observe, 
%\[
%\begin{split}
%P_{t, s}(g, \Gamma) &= E(\mathbbm{1}_{\Gamma}(X_s) | X_t = g) \\
%&= E(E(\mathbbm{1}_{\Gamma}(X_s)|X_u)|X_t=g) \\
%&= E\left(\nu_u(X_u, X_s^{-1}(\Gamma))\middle|X_t = g\right) \\
%&=\int_Sv_u(h, X_s^{-1}(\Gamma))v_t(g, X_u^{-1}(dh)) = \int_{S} P_{t,u}(g,dh)P_{u,s}(h,\Gamma).
%\end{split}
%\]
%
%To show \ref{forward}), take $\epsilon > 0$ and compute, using \ref{chapman}), 
%\[
%\begin{split}
%P_{t, s}(g, \Gamma) &= \int_SP_{t, s-\epsilon}(g, dh) P_{s-\epsilon, s}(h, \Gamma)\\
% &= \int_SP_{t, s-\epsilon}(g, dh) (q_{s-\epsilon}(h, \Gamma) \epsilon +o(\epsilon))
% \end{split}
%\]
%
%\end{proof}
%

%q_{s-\epsilon)(h, \Gamma) \epsilon + o(\epsilon)

\end{lem}
\section{The prospective reserve}

\subsection{Some notions from life insurance}
In this Subsection, we recall some notions from life insurance which can be e.g. found in \cite{Koller} or \cite{Norberg} in the setting of finite-state Markov chains or processes.
\begin{defn}
A (deterministic) \emph{discount function}  $v$ is a continuous function $v:[0, \infty) \times [0, \infty) \rightarrow (0, \infty)$. Often but not always, $v$ derived from a technical interest rate $r$ which is usually positive, i.e. $v(s,t) = e^{-r(t-s)}<1$.

\end{defn}
\begin{defn} 
Given a regular insurance model and a discount function $v$, the \emph{present value of future cashflows} (liabilities) is defined as 
\begin{equation} 
V(t) = \int_{[t, \infty)} v(t,s)\left(dB_{X_s}(s) + b_{X_{s-}, X_s}(s) dN_{X_{s-}, X_s}(s)\right)
\end{equation}
or 

\begin{equation}
\int_{[t, \infty)} v(t, s)dB(s)
\end{equation}

where $B$ is the \emph{total cashflow} up to time $t$ given by

\[
dB_{X_t}(t) + b_{X_{t-}, X_t}(t) dN_{X_{t-}, X_t}(t).
\]

Since $X_t, t\ge0$ is piecewise-constant and cádlág, and the total number of jumps up to time $t$, henceforth called $ N^X(t)$, also is non-decreasing and piecewise constant, both $B_{X_\cdot}$  and $N^X$ are BV (i.e. of bounded variation) and the differentials are interpreted in the sense of Lebesgue-Stieltjes integration. 
\end{defn}
$V$ can in fact be rewritten as: 
\begin{equation}
V(t) = \int_{[t, \infty)}v(t,s)dB_{X_s}(s) + \sum_{i\geq 1} v(J_i)b_{Y_{i-1}, Y_i}(J_i).
\label{altdefv}
\end{equation}
\begin{defn}
The \emph{prospective reserve} is the map $S \times [0, \infty) \ni (g, t) \mapsto V_g(t)$ given by
\[
V_g(t) = E[V(t)|X_t = g].
\]
\end{defn}

\subsection{Thiele's equation}
Using the notation and results of the previous Sections, we are now able to derive the following new type of Thiele's differential equation for the calculation of insurance reserves.
\begin{thm}
\label{ThieleEqThm}
Assume the regular insurance model in Definition \ref{defn1} and suppose that the discount function is of the form $v(t, s) = e^{-\int_t^sr(u)du}$. Then,  
\[
dV_g(t) = (\lambda_t(g) + r(t))V_g(t)dt -dB_g(t) -  \int_S(b_{gh}(t) + V_h(t))q_t(g, dh)dt.
\]

\end{thm}

\begin{proof}
Let $J_t = J_i$ where $i = \min\{k\in \mathbbm{N}|J_k \geq t\}$.
Since 
\begin{equation}
P[J_t\geq u | X_t = g] = P[X_{t+v} = g \hspace{1mm}  \forall v\in (0, u)|X_t = g] = e^{-\int_{[t, u)}\lambda_v(g)dv}
\end{equation}

the real-valued random variable $J_t$ has a conditional density given by 
\begin{equation*}
P(J_t\in ds | X_t = g) = \lambda_s(g)e^{-\int_{[t, s)}\lambda_u(g)du}ds.
\end{equation*}

So 

\begin{eqnarray*}
V_g(t) = E\left[V(t) \middle| X_t = g\right] = \int_{[t, \infty)} P(J_t\in ds | X_t = g)E[V(t)|J_t = s, X_t = g]\\
= \int_{[t, \infty)}\lambda_s(g)e^{-\int_{[t, s)}\lambda_u(g)du}\left\{\int_{[0,s]}v(t, u)dB_g(u) + E[(b_{gY_1}+V_{Y_1}(s))|J_1 = s]\right\}ds \\
=V^{(1)}_g(t) + V^{(2)}_g(t),
\end{eqnarray*}
where
\begin{eqnarray*}
V^{(1)}_g(t) = \int_{[t,\infty)}\lambda_s(g)e^{-\int_{[t, s)}\lambda_u(g)du}\int_{[t,s]}v(t,u)dB_g(u)ds.\\
\end{eqnarray*}

By Fubini's theorem, this is equal to 

\begin{eqnarray*}
\int_{[t,\infty)}v(t,u)\int_{[u,\infty)}\lambda_s(g)e^{-\int_{[t, s)}\lambda_y(g)dy}dsdB_g(u) = \int_{[t, \infty)}v(t,u)e^{-\int_{[t, u)}\lambda_y(g)dy}dB_g(u) \\
=  \int_{[t, \infty)}e^{-\int_{[t, u)}\lambda_y(g)+r(u)dy}dB_g(u).\end{eqnarray*}

This means that  

\begin{eqnarray*}
V^{(2)}_g(t) = \int_{[t,\infty)}v(t, s)\lambda_s(g)e^{-\int_{[t, s)}\lambda_u(g)du} \int_{S\setminus \{g\}} (b_{gh}(s)+ V_h(s)) \pi_s(g, dh)ds\\
= \int_{[t,\infty)}v(t, s)e^{-\int_{[t, s)}\lambda_u(g)du} \int_{S \setminus \{g\}} (b_{gh}(s)+ V_h(s)) q_s(g, dh)ds.
\end{eqnarray*}

Computing the derivatives of $V^{(1)}_g(t)$, resp. $V^{(2)}_g(t)$ with respect to $t$, we get 

\begin{eqnarray*}
dV_g^{(1)}(t) = -dB_g(t) + (\lambda_t(g) + r(t))V^{(1)}_g(t) 
\end{eqnarray*}

and

\begin{eqnarray*}
\frac{d}{dt}V^{(2)}_g(t) = -\int_{S \setminus \{g\}} (b_{gh}(t) + V_h(t)q_t(g, dh) + (\lambda_t(g)+r(t))V^{(2)}_g(t),
\end{eqnarray*}

so

\begin{eqnarray*}
dV_g(t) = (\lambda_t(g) + r(t))V_g(t)dt - dB_g(t) - \int_{S \setminus \{g\}}(b_{gh}(t) + V_h(t)q_t(g, dh).
\end{eqnarray*}
\end{proof}
Finally, we aim at discussing some examples from life insurance which show the need of risk modelling by using Markov processes on more general state spaces.
\begin{example}{(\textbf{The discrete case})}

If $S$ is a countable set with the discrete topology and we assume that $q_t(x, \cdot)$ is a Borel measure on $S$ for every $t\geq 0$ and $x \in S$, the model is reduced to the one described in \cite{Koller}. Values of the $q$-measure on singletons are identical to the jump intensities $(\mu_{ij}(t) = q_t(i, \{j\})$, and the Thiele equation is reduced to the familiar form: 

\[
dV_g(t) = (\lambda_t(g) + r(t))V_g(t)dt -dB_g(t) -  \sum_{h\in S, h \neq g}(b_{gh}(t) + V_h(t))\mu_{gh}(t)dt.
\] 
\end{example}
\begin{example}{(\textbf{Disability insurance with rehabilitation})}

\label{rehabex}
Assume a state space for the insured consisting of three states,
\[
S= \{*, \diamond, \dagger\},
\]

interpreted respectively as healthy, disabled, and deceased. A Markov model with this state space is  an unsatisfactory model for disability insurance, since it implies that the jump intensity from the disabled state to the healthy state (rehabilitation) is solely a function of time. It is clear that a model with any hope of being realistic would have to take into account the dependence of the rehabilitation intensity on the time elapsed since the last transition \emph{into} the disabled state. The solution to this problem (compare to examples in \cite{Koller} in the case of finite-state Markov chains.) lies in replacing $S$ by the state space 
\[
S' \triangleq \{*, \dagger\} \cup (\{\diamond\} \times [0, \infty)),
\]
interpreted as follows: $*$ is the healthy state, $\dagger$ is death, as before. $(\diamond, t)$ means that the insured is disabled and that the jump to disability occurred at the time point $t$. 

The prescription that the rehabilitation intensity $\mu_{\diamond *}(t, \tau)$ should be a  function not only of time $t$, but also of the time $\tau$ since the last jump to $\diamond$, is realized by defining $q$ by 

\[
q_t(*, \{\diamond\} \times H) = \mu_{* \diamond}(t)\mathbbm{1}_H(t),
\]
\[
q_t((\diamond, s), \{*\}) = \mu_{\diamond *}(t, t - s)\mathbbm{1}_{[0, \infty)}(t - s)
\] 

and with the death intensity (from healthy or disabled state) defined in the usual way. 
\end{example}

\begin{example}{(\textbf{Random spouse})}

If we are trying to model an insurance contract giving an annuity payment for a spouse who is left behind when the insured dies, and which is payable continuously until the death of the spouse, the usual route is to employ a two-life model where the state space is a product space of the respective state spaces for the insured and the spouse. However in practice there exist arrangements where the insurer does not know the age of the spouse or whether there even is a spouse at the initiation of the contract, but instead learns of this upon the death of the insured. Nevertheless the insurer has to compute a reserve, which thus has to take into account the random nature of the marital status. In some model definitions, including some mandated as minimum requirements for technical provisions of Norwegian life insurers (see e.g. \cite{capital}), the probability of the existence of a spouse, and the distribution of the age of the spouse, conditional on his or her existence, are given by functions which depend on the age of the insured at the time of death. Since the subsequent evolution of the system involves the mortality of the spouse which is a function of age, and thus depends on the preceding history, Markov chains on a finite or discrete state space are not well-suited for this. The current framework gives a natural resolution to this problem via the following setup: 

\[
S = \{*\} \cup \left(\{\dagger\} \times \mathbbm{R}\right)
\]

where the continuous state variable is interpreted as the age difference between the insured and the spouse revealed at the time of death. Assume given the age-dependent probability $g(t)$  of observing a spouse at the time of death of the insured, and assume that the conditional distribution of the age difference is given by a probability measure $\phi$ on $(\mathbbm{R}, \mathscr{B}(\mathbbm{R}))$. This setup is realized by defining $q$ as follows:

\[
q_t(*, \{\dagger\} \times H)) = \mu_{*\dagger}(t)g(t)\phi(H)
\]

where $\mu_{*\dagger}$ is just the usual mortality rate.
\end{example}

\begin{section}{Computer implementation}
We will give a proof-of-concept numerical implementation method for the life insurance model, focusing on example \ref{rehabex}. 
 Assuming $B_g$ is differentiable with respect to time, we write
\[
\dot{B}_g(t)=\frac{d}{dt}B_g(t).
\]

Theorem \ref{ThieleEqThm} implies 

\[
\frac{d}{dt}V_g(t) = (\lambda_t(g) + r(t))V_g(t)-\dot{B}_g(t)-\int_S(b_{gh}(t)+V_h(t))q_t(g,dh)
\]

In our disability example, we have the below identities with $\delta_g(\cdot)$ denoting the Dirac measure at $g$;

\[
q_t(*, A) = \mu_{*\dagger}(t)\delta_\dagger(A) + \mu_{*\diamond}(t)\delta_{(\diamond, t)}(A)
\]
\[
q_t((\diamond,s),A) = \mu_{\diamond *}(t, t-s)\mathbbm{1}_{[0.\infty)]}(t-s)\delta_*(A) +  \mu_{\diamond \dagger}(t)\delta_\dagger(A).
\]

So

\[
\int_S(b_{*h}(t) + V_h(t))q_t(*, dh) = (b_{*\dagger}(t)+\stackrel{\mathclap{\normalfont\mbox{=0}}}{\overbrace{V_\dagger(t)}})\mu_{*\dagger}(t)+(b_{*,(\diamond,t)}(t)+V_{(\diamond,t)}(t))\mu_{*\diamond}(t).
\]

On the other hand for $g=(\diamond, s)$ we get that 
\[
\int_S(b_{(\diamond,s)h}(t)+V_h(t))q_t((\diamond,s),dh) = (b_{(\diamond,s)*}(t)+V_*(t))\mu_{\diamond *}(t, t-s)\mathbbm{1}_{[0,\infty)}(t-s) + (b_{(\diamond, s)\dagger}(t)+\stackrel{\mathclap{\normalfont\mbox{=0}}}{\overbrace{V_\dagger(t)}})\mu_{\diamond \dagger}(t).
\]

Moreover
\[
\lambda_t(*) \myeq q_t(*, S- \{*\}) = \mu_{*\dagger}(t)+\mu_{*\diamond}(t)
\]
and
\[
\lambda_t((\diamond, s)) \myeq q_t((\diamond,s),S- \{(\diamond,s)\}) = \mu_{\diamond *}(t, t-s)\mathbbm{1}_{[0,\infty)}(t-s)+\mu_{\diamond \dagger}(t).
\]
Hence

\begin{equation}
\label{deq1}
\frac{d}{dt}V_*(t)=(\mu_{*\dagger}(t)+\mu_{*\diamond}(t)+r(t))V_*(t)-\dot{B}_*(t)-\{(b_{*\dagger}(t)+\stackrel{\mathclap{\normalfont\mbox{=0}}}{\overbrace{V_\dagger(t)}})\mu_{*\dagger}(t)+(b_{*, (\diamond,t)}(t)+V_{(\diamond, t)}(t))\mu_{*\diamond}(t)\}, 
\end{equation}
\begin{multline}
\label{deq2}
\frac{d}{dt}V_{(\diamond, s)}(t) = (\mu_{\diamond *}(t,t-s)\mathbbm{1}_{[0,\infty)}(t-s)+\mu_{\diamond \dagger}(t)+r(t))V_{(\diamond, s)}(t),\\ 
-\dot{B}_{(\diamond, s)}(t)-(b_{(\diamond, s)*}(t)+V_*(t))\mu_{\diamond *}(t,t-s)\mathbbm{1}_{[0,\infty)}(t-s)-(b_{(\diamond, s)\dagger}(t)+\stackrel{\mathclap{\normalfont\mbox{=0}}}{\overbrace{V_\dagger(t)}})\mu_{\diamond \dagger}(t).
\end{multline}

The discretized version of (\ref{deq1}) based on an Euler approximation scheme for ordinary differential equations (see \cite{butcher2003}) is the recurrence relation given by 

\begin{multline*}
V_*(t_{n-1})=V_*(t_{n})-(t_{n}-t_{n-1})[(\mu_{*\dagger}(t_{n})+\mu_{*\diamond}(t_{n})+r(t_{n}))V_*(t_{n}) \\ 
-\dot{B}_*(t_{n})-\{(b_{*\dagger}(t_{n}))\mu_{*\dagger}(t_{n})+(b_{*,(\diamond,t_{n})}(t_{n})+V_{(\diamond,t_{n})}(t_{n}))\mu_{*\diamond}(t_{n})\}].
\end{multline*}

Similarly for (\ref{deq2}) and  $k\le n$

\begin{multline*}
V_{(\diamond, t_k)}(t_{n-1}) = V_{(\diamond, t_k)}(t_{n})-(t_{n}-t_{n-1})[(\mu_{\diamond *}(t_{n}, t_{n}-t_k)\mathbbm{1}_{[0,\infty)}(t_{n}-t_k)+\mu_{\diamond \dagger}(t_{n})+r(t_{n}))V_{(\diamond, t_k)}(t_{n})\\
-\dot{B}_{(\diamond, t_k)}(t_{n})-\{(b_{(\diamond, t_k)*}(t_{n})+V_*(t_{n}))\mu_{\diamond *}(t_{n},t_{n}-t_k)+b_{(\diamond, t_k)}(t_{n})\mu_{\diamond \dagger}(t_{n})\}.
\end{multline*}

We will implement this assuming that the disability insurance pays 1\$ per year as long as the insured is in the disabled state, but only until the age of retirement which we set to 67 years. We also assume a constant force of interest $r$. The recursion scheme is then reduced to 

\begin{multline}
V_*(t_{n-1})=V_*(t_{n})-(t_{n}-t_{n-1})[(\mu_{*\dagger}(t_{n})+\mu_{*\diamond}(t_{n})+r)V_*(t_{n}) 
-V_{(\diamond,t_{n})}(t_{n})\mu_{*\diamond}(t_{n})],
\end{multline}

\begin{multline}
V_{(\diamond, t_k)}(t_{n-1}) = V_{(\diamond, t_k)}(t_{n})-(t_{n}-t_{n-1})[(\mu_{\diamond *}(t_{n}, t_{n}-t_k)+\mu_{\diamond \dagger}(t_{n})+r)V_{(\diamond, t_k)}(t_{n})\\
-1-V_*(t_{n})\mu_{\diamond *}(t_{n},t_{n}-t_k)]
\end{multline}

with the boundary condition
\[
V_g(67)=0,  g\in S.
\]

Here we use transition rates of Gompertz-Makeham type, except for the rehabilitation rate which is somewhat more involved.
Specifically we set

\[
\mu_{*\dagger}(t)=0.0004+10^{0.060t-5.46},
\]
\[
\mu_{*\diamond}(t) = 0.0005+10^{0.038t-4.12}.
\]
Further, we model the two-parameter continuous-time rehabilitation rate as follows:
\[
\mu_{\diamond *}(t,\tau) = \mu_{\diamond *}(t)(1-\mu_{\diamond \dagger}(t_0+s))\mathbbm{1}_{[0,\infty)}(t-s),
\]

where

\[
\mu_{\diamond *}(t) = 0.773763-0.01045t
\]

and $t_0$ is a parameter to be adjusted to the specific situation. In our example $t_0$ is given by the age of the insured at the start of the contract. 
\subsection{Results}
Displayed in Figures \ref{fig1} and \ref{fig2} are the reserve plots of the case illustrated in example \ref{rehabex}. We observe the prospective reserve in the active state declines with increasing age of the insured in Figure \ref{fig1}. In Figure \ref{fig2}, we note the prospective reserve in the disabled state of the insured is significantly impacted. In other words, the prospective reserve in the disabled state increases with the age of the insured from onset of disability.
\par
Comparing our result with that of \emph{Example 2.4.2} in \cite{Koller}, it is obvious our model outperforms the classical model by incorporating the influence of rehabilitation on the disabled state of the insured for the computation of prospective reserves. This application justifies the use of Markov processes with non-countable state spaces in connection with our proposed model.
\begin{figure}[!htbp]
 \adjustimage{width=1.2\textwidth,center}{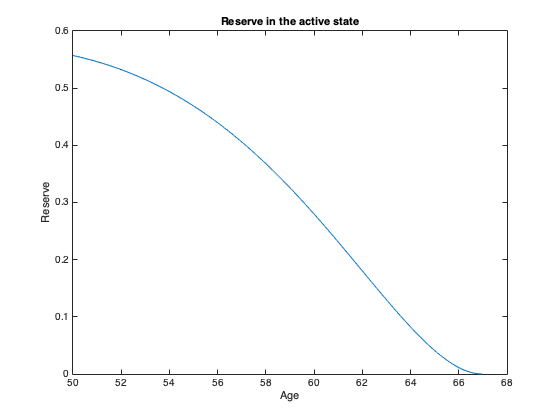}
 \caption{Plot of reserve in the active state}
 \label{fig1}
\end{figure}

\begin{figure}[!htbp]
 \adjustimage{width=1.2\textwidth,center}{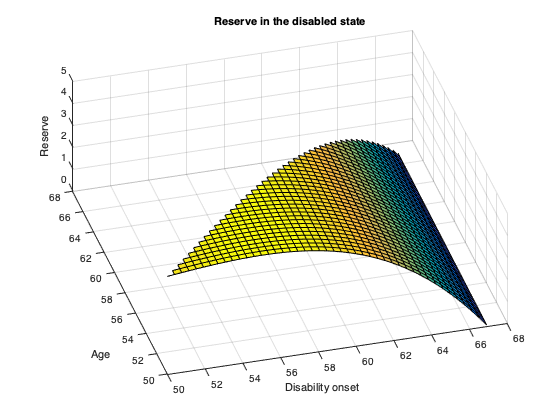}
 \caption{Plot of reserve in the disabled state}
 \label{fig2}
\end{figure}
\subsection*{Acknowledgements}
One of the authors (S.D.) wants to thank Tor Vidvei for interesting discussions on application problems which led to the idea behind this paper.

\end{section}

\bibliography{mybib}{}
\bibliographystyle{alpha}

\end{document}